\documentclass[aps,prd,10pt,notitlepage,nofootinbib,superscriptaddress,showkeys,showpacs]{revtex4-1}

\usepackage{amsmath,amssymb,amsthm,latexsym, bbm}

\usepackage[english]{babel}
\usepackage{color}
\usepackage{xspace}
\usepackage{graphicx}
\usepackage{pifont,dsfont}
\usepackage{marvosym}
\usepackage{slashed}
\usepackage[caption=false]{subfig}
\usepackage{enumitem}

\usepackage{hyperref}

\theoremstyle{definition}
\newtheorem{lemma}{Lemma}

\newtheorem{definition}{Definition}
\newtheorem{theorem}{Theorem}

\newtheorem{corollary}{Corollary}
\newtheorem{proposition}{Proposition}

\newcommand{\cG}{\mathcal{G}}

\makeatletter
\renewcommand\paragraph{\@startsection{paragraph}{4}{\z@}%
                                     {3ex\@plus 1ex \@minus -.2ex}%
                                     {-1ex \@plus .2ex}%
                                     {\normalfont\normalsize\bfseries}}
\makeatother

\begin{document}

\title{\Large \bf Another proof of the $1/N$ expansion of the rank three tensor model with tetrahedral interaction}

\author{{\bf Valentin Bonzom}}\email{bonzom@lipn.univ-paris13.fr}
\affiliation{LIPN, UMR CNRS 7030, Institut Galil\'ee, Universit\'e Paris 13,
99, avenue Jean-Baptiste Cl\'ement, 93430 Villetaneuse, France, EU}
\affiliation{IRIF, UMR CNRS 8243, Universit\'e de Paris, France, EU}

\date{\small\today}

\begin{abstract}
\noindent The rank three tensor model with tetrahedral interaction was shown by Carrozza and Tanasa to admit a $1/N$ expansion, dominated by melonic diagrams, and double tadpoles decorated with melons at next-to-leading order. This model has generated a renewed interest in tensor models because it has the same large $N$ limit as the SYK model. In contrast with matrix models, there is no method which would be able to prove the existence of $1/N$ expansions in arbitrary tensor models. The method used by Carrozza and Tanasa proves the existence of the $1/N$ expansion using two-dimensional topology, before identifying the leading order and next-to-leading graphs. However, another method was required for complex, rank three tensor models with planar interactions, which is based on flips. The latter are moves which cut two propagators of Feynman graphs and reglue them differently. They allow transforming graphs while tracking their orders in the $1/N$ expansion. Here we use this method to re-prove the results of Carrozza and Tanasa, thereby proving the existence of the $1/N$ expansion, the melonic dominance at leading order and the melon-decorated double tadpoles at next-to-leading order, all in one go.

\end{abstract}

\medskip


\maketitle

\section*{Introduction}

\paragraph{The $1/N$ expansions of tensor models --} Tensor models are natural generalizations of matrix models. It was however necessary to wait up until 2010 to find a $1/N$ expansion in tensor models \cite{1/NExpansion}. For matrix models \cite{MatrixReview}, the $1/N$ expansion can be derived by observing that the Feynman graphs are discrete surfaces which are classified by their genera. In tensor models however, the combinatorics of the Feynman graphs is much more intricate. The relationship to discrete manifolds also is more subtle \cite{GurauBook}. Not every tensor models has Feynman graphs corresponding to discretized (pseudo-)manifolds, and even when it has, it is not as helful as in matrix models since those manifolds are of higher dimensions and are thus not as easily classified as surfaces.

As we discuss in this introduction and in contrast this matrix models, there are various methods in the literature which can be used to prove the existence of $1/N$ expansions for tensor models but none of them works for all models, which makes the situation not fully satisfactory.

Since the fields in tensor models have more indices than matrices, more interactions and observables can be constructed. The set of interactions and observables in fact depends on the symmetries of the model. A useful representation is to represent each $T$ of the observable or interaction as a vertex and draw an edge between two vertices when there is an index identification between the corresponding tensors. Such a drawing is called a \emph{bubble}.

The first $1/N$ expansion was found for complex tensors having no particular symmetry among their indices. In that case \cite{Uncoloring}, the natural invariance of the observables and interactions is $U(N)^d$, for complex tensors of rank $d$. The multiorientable model \cite{MO, MO-Review} in rank three has as invariance $U(N)\times O(N)\times U(N)$ and a quartic interaction and came shortly after. Carrozza and Tanasa \cite{CarrozzaTanasa} then extended it to real tensors with a tetrahedral interaction invariant under $O(N)^3$,
\begin{equation} \label{K4Bubble}
K_4(T) = \sum_{\substack{a_1, a_2, a_3\\ b_1, b_2, b_3}} T_{a_1 a_2 a_3} T_{a_1 b_2 b_3} T_{b_1 b_2 a_3} T_{b_1 a_2 b_3} = \begin{array}{c} \includegraphics[scale=.4]{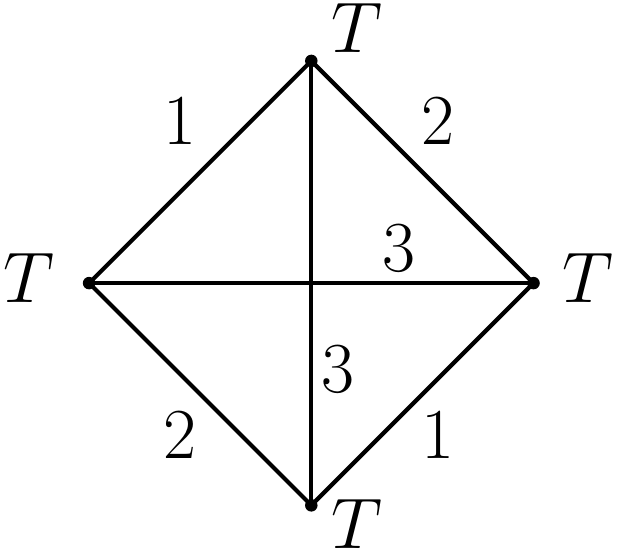} \end{array}
\end{equation}
The corresponding bubble is the graph of the tetrahedron $K_4$ and the edge labels indicate the positions of the indices which are contracted.

The free energy of this model is
\begin{equation}
F_N(t) = \ln \int dT\ \exp -N^2\sum_{a_1, a_2, a_3} T_{a_1 a_2 a_3}T_{a_1 a_2 a_3} - N^{5/2}\,t\,K_4(T)
\end{equation}
and it admits an expansion over connected Feynman graphs as follows. A Feynman graph has some bubbles, all copies of $K_4$, and propagators connecting vertices pairwise. As usual, propagators will be given the fictitious color 0 and represented as dashed edges. For instance, there are three double tadpoles with one bubble,
\begin{equation} \label{GraphsEx}
G^{(1)}_1 = \begin{array}{c} \includegraphics[scale=.4]{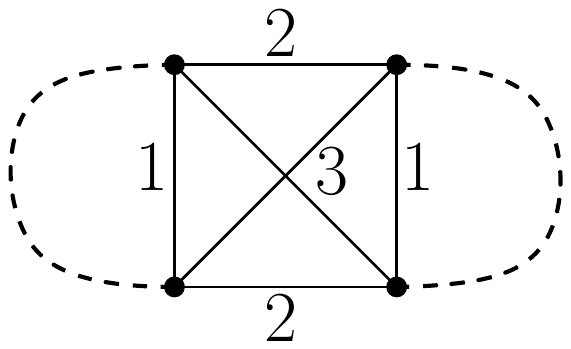} \end{array} \qquad G^{(2)}_1 = \begin{array}{c} \includegraphics[scale=.4]{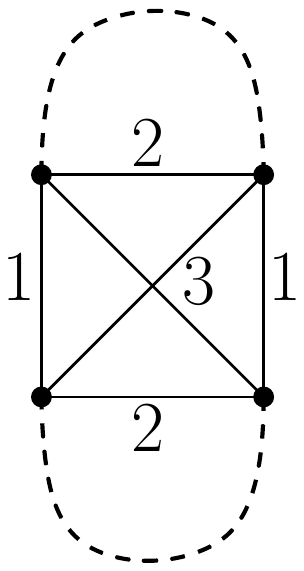} \end{array} \qquad G^{(3)}_1 = \begin{array}{c} \includegraphics[scale=.4]{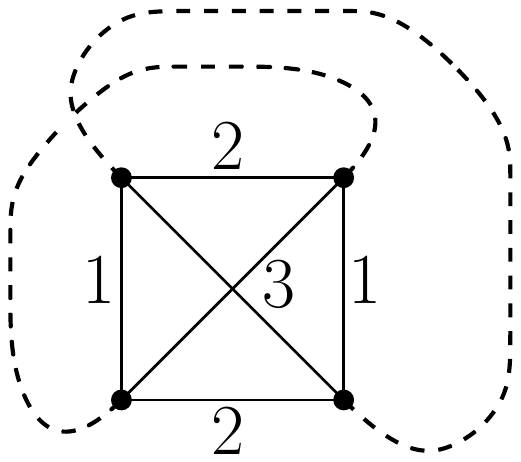} \end{array}
\end{equation}
and the leading order Feynman graph with two bubbles is
\begin{equation}
G_2 = \begin{array}{c} \includegraphics[scale=.4]{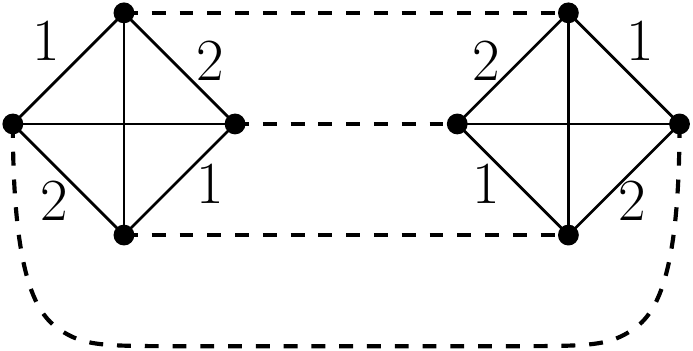} \end{array}
\end{equation}

In all cases mentioned above, the strategy was to first establish the existence of a $1/N$ expansion, i.e. writing the free energy as
\begin{equation}
F_N(t) = \sum_{\omega \geq 0} N^{d-\omega} F^{(\omega)}(t)
\end{equation}
which is done by identifying some sub-surfaces $\{J\}$, called jackets, in the Feynman graphs such that $\omega = \sum_J g_J$ where $g_J$ is the genus of $J$. This implies that $\omega \geq 0$ and the $1/N$ expansion exists. 

The second step was to identify the graphs which dominate the large $N$ limit, i.e. $\omega =0$, and which turn out the be the special class of melonic graphs\footnote{The terminology of melonic graphs might be a little confusing because they refer to graphs from several models and thus built from different bubbles. The idea is that they have the same structure. The same goes with the notion of melonic bubbles. The relevant definitions will be given in the next section.}, and then classify the graphs which appear at the subsequent orders. We give a few theorems and references.

\begin{theorem}
Tensor models with melonic interactions and $U(N)^d$ invariance admit a $1/N$ expansion dominated by melonic graphs \cite{Melons, Uncoloring}. The graphs at fixed $\omega$ are classified in \cite{Gurau-Schaeffer}. 

Moreover, the large $N$ limit is Gaussian, in the sense that the large $N$ limits of the expectations of all $U(N)^d$-invariant observables are those of the Gaussian tensor with covariance given by the large $N$ 2-point function \cite{Universality}.
\end{theorem}


\begin{theorem} \label{thm:O(N)Model} \cite{CarrozzaTanasa}
The rank three tensor model with tetrahedral interaction and $O(N)^3$ invariance admits a $1/N$ expansion dominated by melonic graphs. The NLO is given by the three double tadpoles decorated with melonic 2-point functions. 
\end{theorem}

In fact, more is known about this model, but we have kept the theorem as it is because it is the one proved by Carrozza and Tanasa and the one we will revisit here. The NNLO and NNNLO were found in \cite{NadorCTKT}. There, it was also proved that there is a bijection between the graphs of this model which have an orientable jacket and those of the multiorientable model \cite{NadorCTKT}. The classification of the graphs of the multiorientable model at fixed $\omega$ is given in \cite{FusyTanasaMO}.

\begin{theorem} \cite{FerrariRivasseauValette}
$O(N)^d$-invariant models with maximally symmetric interactions admit $1/N$ expansions. When the corresponding bubbles are the complete graphs on $R+1$ vertices with $R$ prime, the large $N$ limit is dominated by mirror melons.
\end{theorem}

Notice that the latter theorem includes the $K_4$ interaction, in rank three, where mirror melons reduce to melonic graphs.

The study of $U(N)^d$-invariant tensor models with non-melonic interactions started with \cite{New1/N} and \cite{MelonoPlanar} in $d\geq 4$, which found that non-trivial\footnote{Here non-trivial large $N$ limit means that an infinite number of graphs contribute at large $N$.} large $N$ limits exist, which can lead to different phases: melonic (i.e. branched polymers in the continuum), planar ribbon graphs (2D quantum gravity phase) and a phase transition between them (this is the same phase portrait as in matrix models with multi-trace interactions). The corresponding multicritical behaviors were given in \cite{LionniThurigen}.

These $U(N)^d$-invariant models with non-melonic interactions were not solved with the same method as the theorems outlined above, but instead using intermediate field methods, i.e. bijections on the Feynman graphs, and often proving that a $1/N$ expansion exists and which the dominant graphs are at large $N$ both in one go. This typically involves \emph{moves} to transform an arbitrary graph into another (which has a feature of the leading order terms) and compare their exponents of $N$.

The case of rank three, $U(N)^3$-invariant tensors with non-melonic interactions is more difficult. The same bijective methods \cite{LucaPhD} were used on a couple of examples in \cite{IF-LionniRivasseauBonzom}, with the (non-planar) complete bipartite graph $K_{3,3}$ as interaction bubble, and in \cite{Octahedra} with a (planar) cube as interaction bubble. The latest result \cite{PlanarBubbles} is thus an important extension to more generic interactions, and as it turns out extends Gurau's universality theorem for large random tensors \cite{Universality}.

\begin{theorem} \cite{PlanarBubbles}
$U(N)^3$-invariant tensor models with any planar bubbles as interactions admit $1/N$ expansions and their large $N$ limits are Gaussian.
\end{theorem}

This theorem does not use any bijection and solely relies on \emph{local moves} which are called \emph{flips}. They consist in cutting two propagators which are incident to the same bubble (hence ``local'') and regluing them differently,
\begin{equation}
\begin{array}{c} \includegraphics[scale=.4]{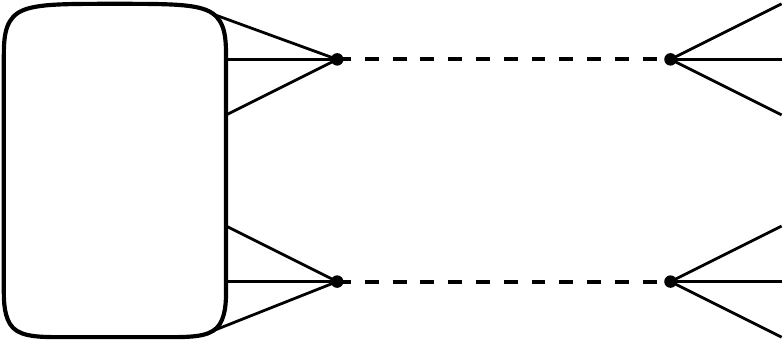} \end{array} \qquad \to \qquad\begin{array}{c} \includegraphics[scale=.4]{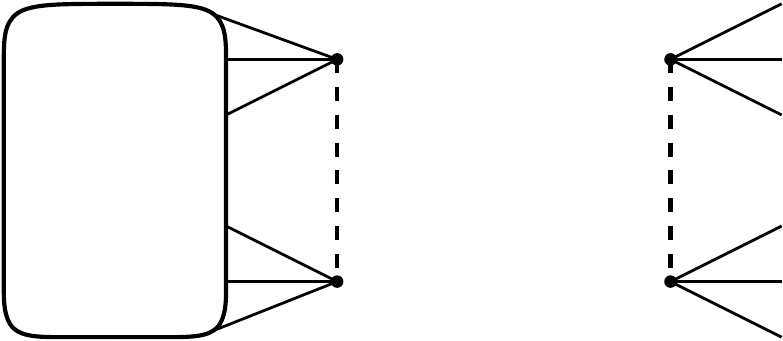} \end{array}
\end{equation}
This method proves in one go the existence of the $1/N$ expansion and the large $N$ Gaussianity, by identifying the Feynman graphs which maximize the exponents of $N$ at fixed number of interactions. 

\paragraph{Our result: a new proof of Theorem \ref{thm:O(N)Model} --} In this paper we simply revisit Theorem \ref{thm:O(N)Model} on the $K_4$ model: we prove the existence of its $1/N$ expansion, the melonic dominance and the NLO, all at once. We do so using the same flips as in \cite{PlanarBubbles}. There is nevertheless a difference: the flips will not be as local, since they will act on propagators connected to bubbles which are adjacent (instead of the same bubble). In particular, our method does not use the combinatorics of any sub-surfaces such as the jackets and their 2-dimensional topology.

\paragraph{Motivations --} Tensor models have attracted some attention lately beyond the quantum gravity, matrix models and group field theory communities, as it was noted by Witten that (colored, i.e. multi-tensor) tensorial quantum-mechanical models have the same large $N$ limit as the SYK model with the advantage of not having disorder. It was then noticed how to use single-tensor models, in particular the $K_4$ model in \cite{KlebanovTarnopolsky, FerrariLargeD}. This also revived interests in multi-matrix models \cite{Azeyanagi} with new large $N$, large $D$ behaviors and sparked the development of tensor models with less symmetry. In particular, $O(N)$-invariant models \cite{O(N)1, O(N)2} and $Sp(N)$-invariant models \cite{Sp(N)} have been shown to possess $1/N$ expansions. These late developments have however required yet other techniques, in particular because cancellations of diagrams with divergent exponents are involved.

The reader can thus see that there are no known unifying techniques which could be used to prove the existence of $1/N$ expansions in any type of tensor models. We therefore think that in spite of us revisiting a known result, the present paper shows that the method developed in \cite{PlanarBubbles} is not just a one-off and has potential applicability for other models.

As it turns out, the combinatorial proof of the melonic dominance in the SYK model \cite{MelonsSYK} also works similarly: it proves the existence of the large $N$ limit and the melonic dominance in one go and does so using only moves (no 2-dimensional topology involved). The moves are different because the graphs involve different building blocks but it seems reasonable to put this proof in the same bag as \cite{PlanarBubbles} and the present paper.

\section{The tetrahedral tensor model, its large $N$ limit and next-to-leading order}

Let $\cG$ be the set of connected Feynman graphs of the $K_4$ model, i.e. graphs obtained by connecting some $K_4$ bubbles, with colored edges as in \eqref{K4Bubble}, with edges of color 0 between vertices, such that every vertex is incident to exactly one edge of color 0. For $G\in\cG$, we denote $b(G)$ the number of bubbles of $G$. The number of vertices is then $4b(G)$ and the number of edges of color 0 is $2b(G)$.

\begin{definition}
A \emph{face} of color $c\in\{1,2,3\}$ is a cycle alternating edges of color 0 and $c$. The \emph{length} of a face is the number of edges of color 0 of the cycle.

We say that a face is \emph{proper} if it never visits the same bubble more than once. In other words, a proper face of length $n$ visits exactly $n$ distinct bubbles.
\end{definition}

The terminology of ``face'' is inherited from matrix models where the same type of cycles are the boundaries of the faces of the surfaces (the faces themselves being the connected components of the complement of the graph). For instance, the double tadpole $G_1^{(c)}$ has 2 faces of color $c$ and 1 face for the two other colors. The graph $G_2$ has two faces of each color, each being a proper face of length 2.

We will only use the notion of proper face for faces of length 2. Indeed, a face of length 2 can visit the same bubble twice, in which case the graph is necessarily a double tadpole $G_1^{(c)}$ of \eqref{GraphsEx}, or it can visit two different bubbles. We consider proper faces to get rid of the double tadpoles.

The exponent of $N$ associated to $G\in\cG$ is calculated as follows. There is a factor $1/N^2$ for each edge of color 0 (the bare propagator) and a factor $N^{5/2}$ for each bubble. Moreover, since the propagator simply identifies the indices which are in the same position, one gets a power of $N$ for each face, so that
\begin{equation}
F_N(t) = \sum_{G\in\cG} s(G)\ N^{F(G)-\frac{3}{2}b(G)}\ (-t)^{b(G)}
\end{equation}
where $F(G)$ is the total number of faces and $s(G)$ is a numerical factor.

The strategy in \cite{CarrozzaTanasa} is as follows.
\begin{itemize}
\item Proving the existence of the $1/N$ expansion, i.e. proving that the exponent of $N$ is bounded,
\begin{equation}
\exists\, d \qquad \forall\, G\in\cG \qquad F(G)-\frac{3}{2}b(G) \leq d
\end{equation}
\item Identify the graphs which reach the bound (and resum the free energy which is here trivial).
\end{itemize}

Here we will do both at once. Introduce $\cG_{\max}\subset \cG$ the set of graphs which maximize the number of faces at fixed $b(G)$, for all values of $b(G)$. Our job will be to identify the graphs of $\cG_{\max}$ and evaluate their number of faces in term of $b(G)$.

As often in tensor models, 2-point functions play a special role. 
\begin{definition}
A \emph{2-point function} (or 2-point graph) is a graph $G^\circ$ which is either just an edge of color 0, or is 
obtained from a graph $G\in \cG$ by cutting an edge of color 0 into two disconnected halves. In the latter case, we say that it is a non-trivial 2-point function. A 2-point function will be graphically represented with a disk on an edge of color 0.
\end{definition}

A graph $G\in\cG$ which contains a non-trivial 2-point function is said to be \emph{2-particle reducible} (2PR). Equivalently, it contains a \emph{2-edge-cut}, i.e. a pair of edges $\{e_L, e_R\}$ whose removal disconnects $G$ into two connected components. In particular, we say that two vertices are connected by a 2-point function if they are connected by an edge of color 0 or if they are incident to two edges of color 0 forming a 2-edge-cut.

Next we define the melonic graphs for the $K_4$ model, and its melon-decorated double tadpoles.

\begin{definition}
The \emph{melonic dipole} is the 2-point graph obtained by cutting an edge of color 0 in $G_2$,
\begin{equation}
\begin{array}{c} \includegraphics[scale=.4]{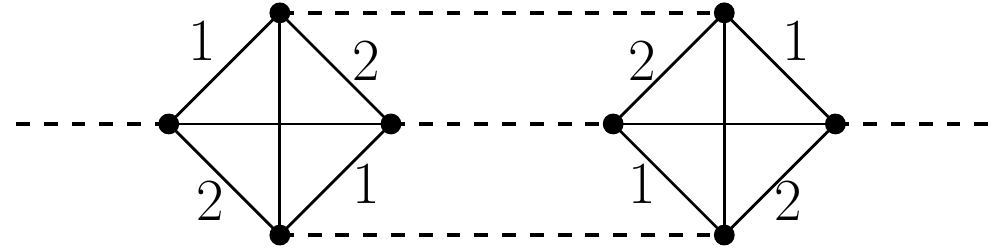} \end{array}
\end{equation}

Melonic graphs are then defined inductively. They only exist for $b(G)\in 2\mathbb{N}$. The only melonic graph with two bubbles is $G_2$. A melonic graph with $b+2$ bubbles is obtained by inserting the melonic dipole on any edge of color 0 of a melonic graph with $b$ bubbles.

The double tadpoles decorated with melons are any graph of the form 
\begin{equation}
\begin{array}{c} \includegraphics[scale=.4]{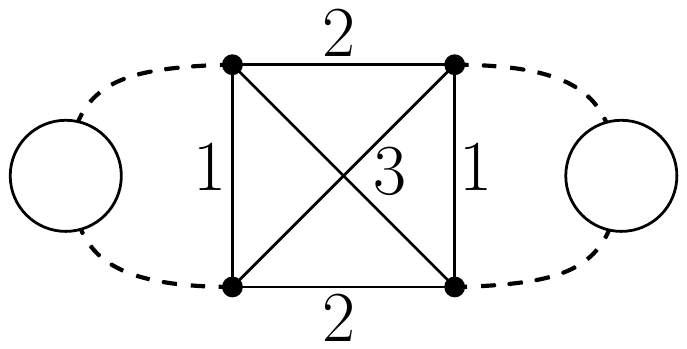} \end{array} \qquad
\begin{array}{c} \includegraphics[scale=.4]{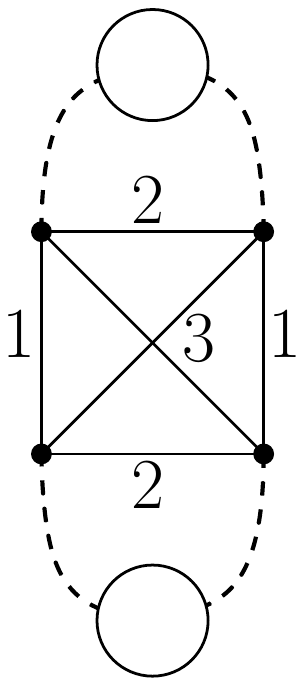} \end{array} \qquad
\begin{array}{c} \includegraphics[scale=.4]{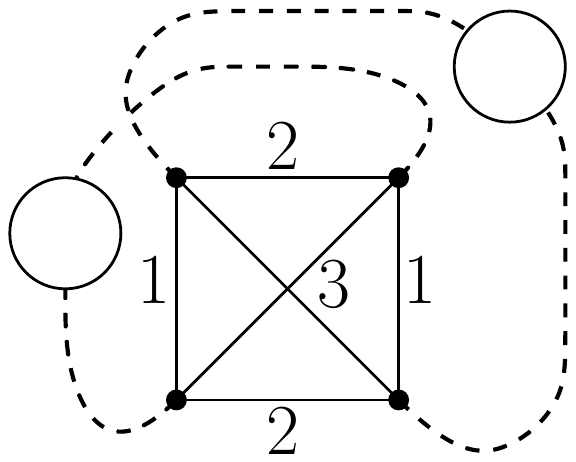} \end{array}
\end{equation}
where the blobs are \emph{melonic} 2-point functions.
\end{definition}

The following theorem is simply a reformulation of Theorem \ref{thm:O(N)Model} from \cite{CarrozzaTanasa}.

\begin{theorem} \label{thm:Gmax}
$\cG_{\max} = \cG_{\max}^{\text{odd}}\sqcup \cG_{\max}^{\text{even}}$, where 
\begin{itemize}
\item $\cG_{\max}^{\text{even}}$ is the set of all melonic graphs. They maximize the number of faces at even number of bubbles and $F(G) = 3b(G)/2+3$.
\item $\cG_{\max}^{\text{odd}}$ is the set of graphs made of a double tadpole decorated with melons. They maximize the number of faces at odd number of bubbles and $F(G) = 3b(G)/2 + 5/2$.
\end{itemize}
\end{theorem}

This accounts for the existence of the $1/N$ expansion, the melonic dominance at large $N$ and the NLO. The remaining of the paper is devoted to the proof.

\section{Proof}

\subsection{Number of faces of melonic graphs and the double tadpole}

The easiest part of the Theorem is the number of faces of the relevant graphs. This will also be used later in the proof of the rest of the theorem.

\begin{proposition} \label{thm:NumberFaces}
The number of faces of melonic graphs with $b$ bubbles is
\begin{equation}
F(G) = \frac{3}{2}b + 3.
\end{equation}
The number of faces of a melon-decorated double tadpole with $b$ bubbles is
\begin{equation}
F(G) = \frac{3}{2}(b-1) + 4.
\end{equation}
\end{proposition}

\begin{proof}
Those formulae result from elementary inductions. A melonic dipole adds two bubbles and three faces, which explains the coefficient $3/2$. The offset is found by observing that the melonic graph $G_2$ with two bubbles has six faces and a double tadpole $G_1^{(c)}$ has four faces.
\end{proof}

\subsection{Flips}

Flips are the most essential tools of the proof.

\begin{definition}
Consider two edges of color 0 $e_L, e_R$ in $G$. A \emph{flip} of $\{e_L, e_R\}$ is the transformation from $G$ to $G'$ as follows
\begin{equation}
\begin{array}{c} \includegraphics[scale=.5]{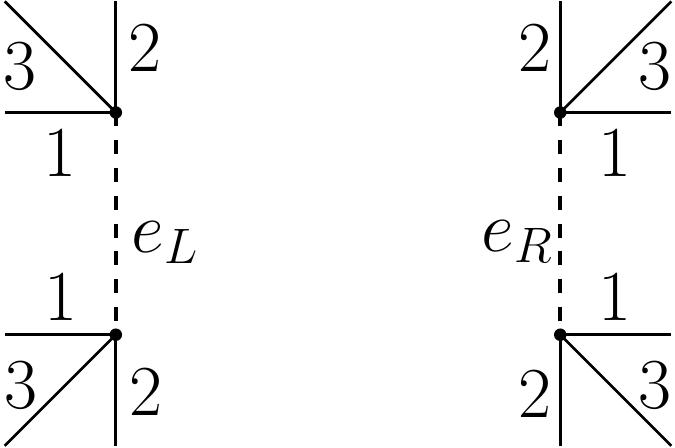} \end{array} \qquad \to \qquad  \begin{array}{c} \includegraphics[scale=.5]{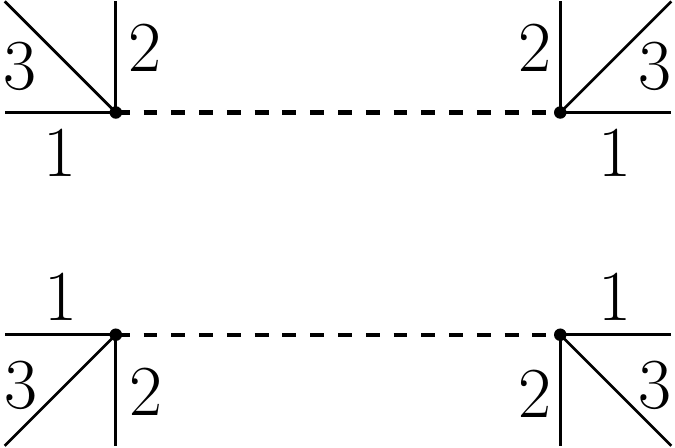} \end{array} \qquad \text{or} \qquad \begin{array}{c} \includegraphics[scale=.5]{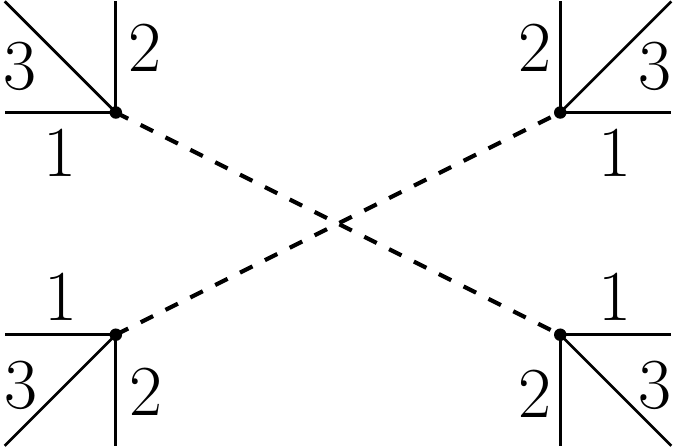} \end{array}
\end{equation}
\end{definition}

\begin{lemma} \label{thm:DeltaF}
Consider that $\{e_L, e_R\}$ is not a 2-edge-cut. Denote $\Delta_{G\to G'} F_c = F_c(G') - F_c(G)$ the variation of the number of faces of color $c\in\{1, 2, 3\}$ in a flip of $\{e_L, e_R\}$. Then $|\Delta_{G\to G'} F_c|\leq 1$.

Moreover, denote $F_c^{\{e_L, e_R\}}(G)$ the number of faces of color $c$ which go along $e_L$ or $e_R$. Then
\begin{itemize}
\item if $F_c^{\{e_L, e_R\}}(G) = 2$, then $\Delta_{G\to G'} F_c=-1$
\item if $F_c^{\{e_L, e_R\}}(G) = 1$, then there exists a flip such that $\Delta_{G\to G'} F_c=1$ while the other gives $\Delta_{G\to G'} F_c=0$.
\end{itemize}
\end{lemma}

\begin{proof}
Either it is the same face of color $c$ along $e_L$ and $e_R$, meaning $F_c^{\{e_L, e_R\}}(G) = 1$, or they are two different faces i.e. $F_c^{\{e_L, e_R\}}(G) = 2$ .

Let us set $c=1$ for concreteness. The case $F_c^{\{e_L, e_R\}}(G) = 2$ means that there are two paths of colors $\{0,1\}$ as follows
\begin{equation}
\begin{array}{c} \includegraphics[scale=.4]{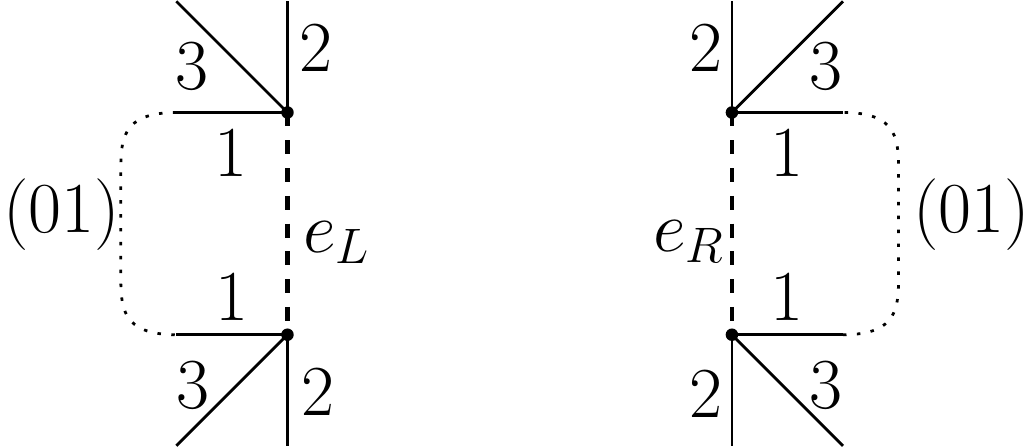} \end{array}
\end{equation}
Clearly after any of the two possible flips those two faces merge meaning $\Delta_{G\to G'} F_c=-1$.

In the case $F_c^{\{e_L, e_R\}}(G) = 1$, the situation can be 
\begin{equation}
\begin{array}{c} \includegraphics[scale=.4]{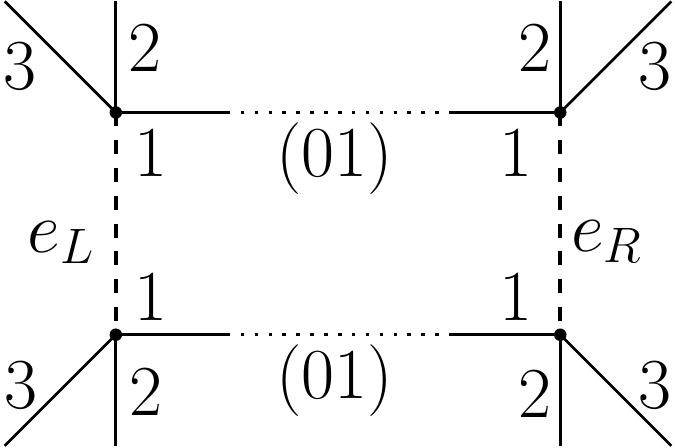} \end{array} \qquad \text{or} \qquad 
\begin{array}{c} \includegraphics[scale=.4]{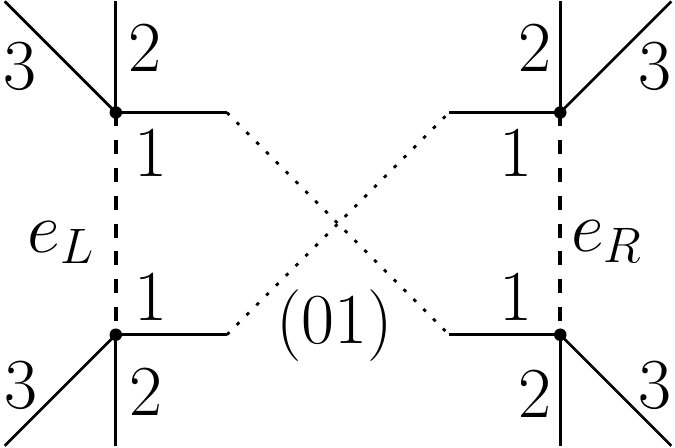} \end{array}
\end{equation}
and for both it is immediate to see that there is a flip which splits the face into two while the other flip does not change the number of faces.
\end{proof}

\begin{corollary} \label{thm:DeltaFTotal}
If there are two colors $c_1, c_2$ such that $F_{c_1}^{\{e_L, e_R\}}(G) = F_{c_2}^{\{e_L, e_R\}}(G)=1$, then there exists a flip such that
\begin{equation}
\Delta_{G\to G'} F =\sum_{c\in\{1,2,3\}} \Delta_{G\to G'} F_c >0
\end{equation}
except possibly $\Delta_{G\to G'} F=0$ if $F_{c_3}^{\{e_L, e_R\}}(G)=2$ for the third color $c_3$.
\end{corollary}

\begin{proof}
From the above lemma, if $F_{c}^{\{e_L, e_R\}}=1$ for all colors, then there exists a flip such that $\Delta_{G\to G'} F_{c_1}>0$ while $\Delta_{G\to G'} F_{c_2}, \Delta_{G\to G'} F_{c_3} \geq 0$. If however $F_{c_3}^{\{e_L, e_R\}}=2$ then the lemma only guarantees $\Delta_{G\to G'} F\geq 0$.
\end{proof}

\subsection{Proper faces of length 2}

\begin{lemma} \label{thm:Length2}
If $G\in\cG_{\max}$ has a proper face of length 2 and color 1, then it has the form
\begin{equation}
G = \begin{array}{c} \includegraphics[scale=.4]{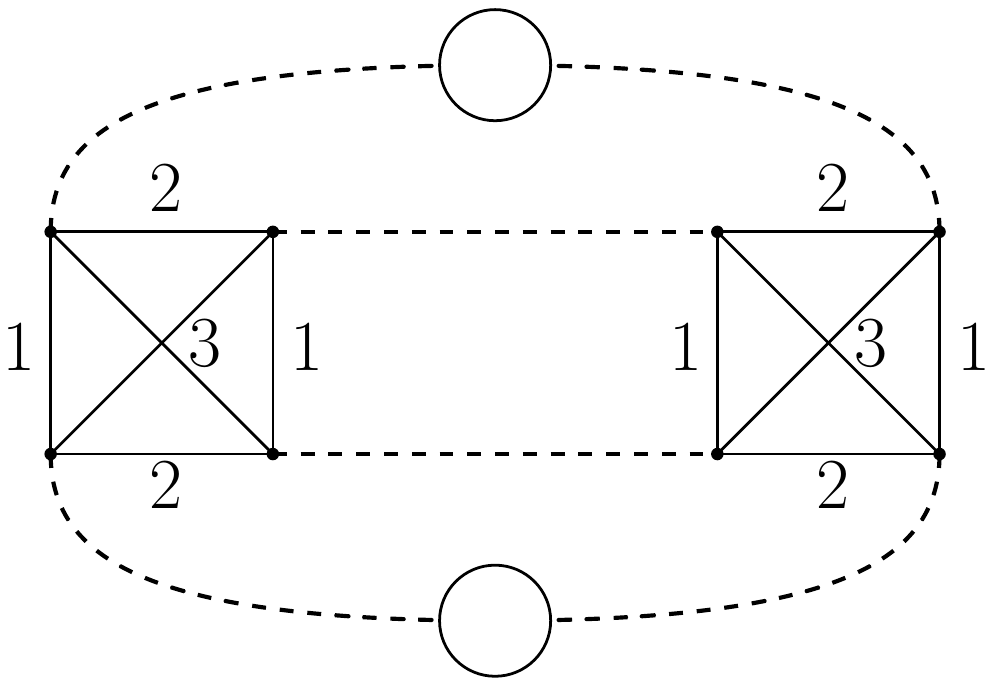} \end{array}
\end{equation}
and the same result holds with permuted colors.
\end{lemma}

\begin{proof}
$G$ has a proper face of length 2 and color 1, thus contains a portion like
\begin{equation}
\begin{array}{c} \includegraphics[scale=.4]{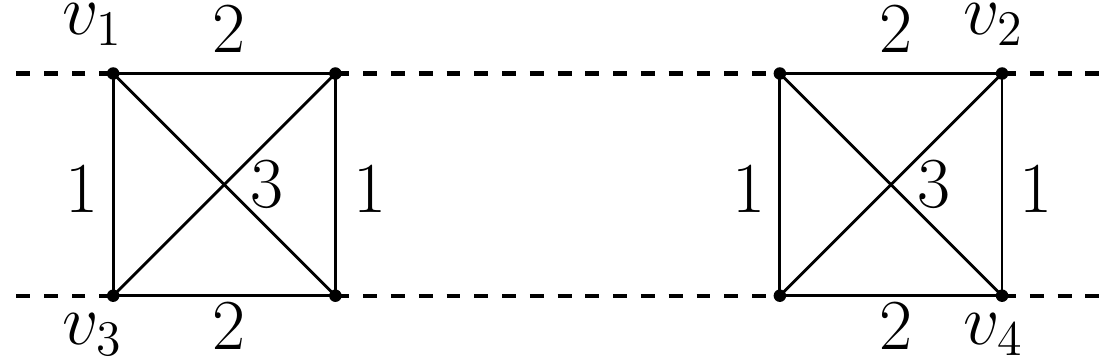} \end{array}
\end{equation}
and we will show that there must be a 2-point function between $v_1$ and $v_2$ and between $v_3$ and $v_4$.

Notice that if there is a 2-point function between $v_1$ and $v_2$ then there is also one between $v_3$ and $v_4$ (and the other way around). We therefore assume that there are no such 2-point functions, so that the edges incident to $v_1 $ and $v_2$ are distinct and do not form a 2-edge-cut. The situation is as follows
\begin{equation}
\begin{array}{c} \includegraphics[scale=.4]{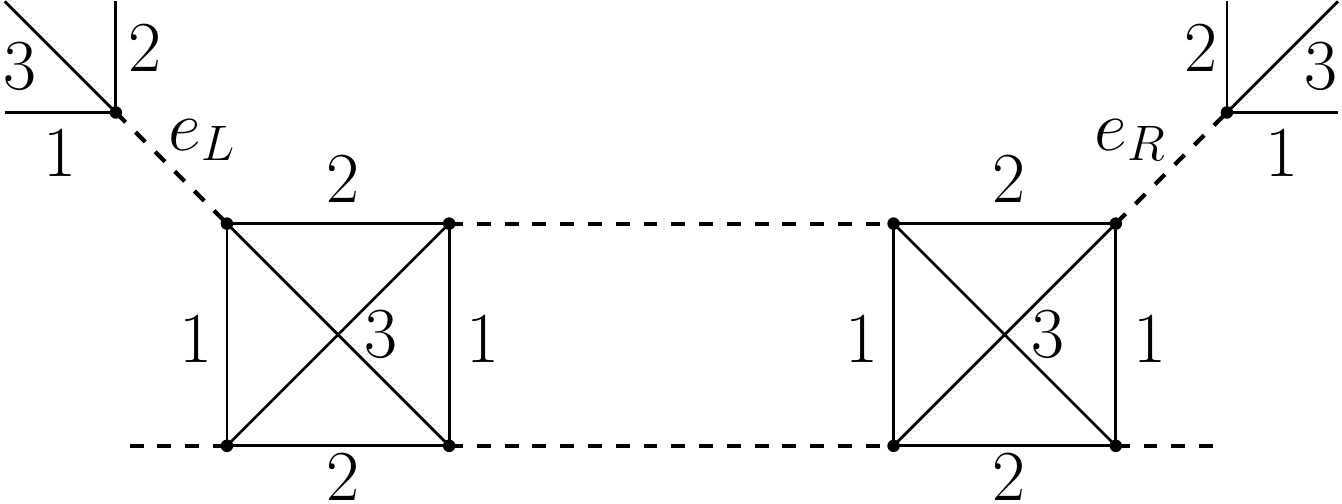} \end{array}
\end{equation}
Notice that $F^{\{e_L, e_R\}}_2 = F^{\{e_L, e_R\}}_3 =1$, therefore from Lemma \ref{thm:DeltaF} there is a flip which increases the number of faces of color 2, and same for the color 3. As it turns out it is the same flip for both colors,
\begin{equation}
\begin{array}{c} \includegraphics[scale=.4]{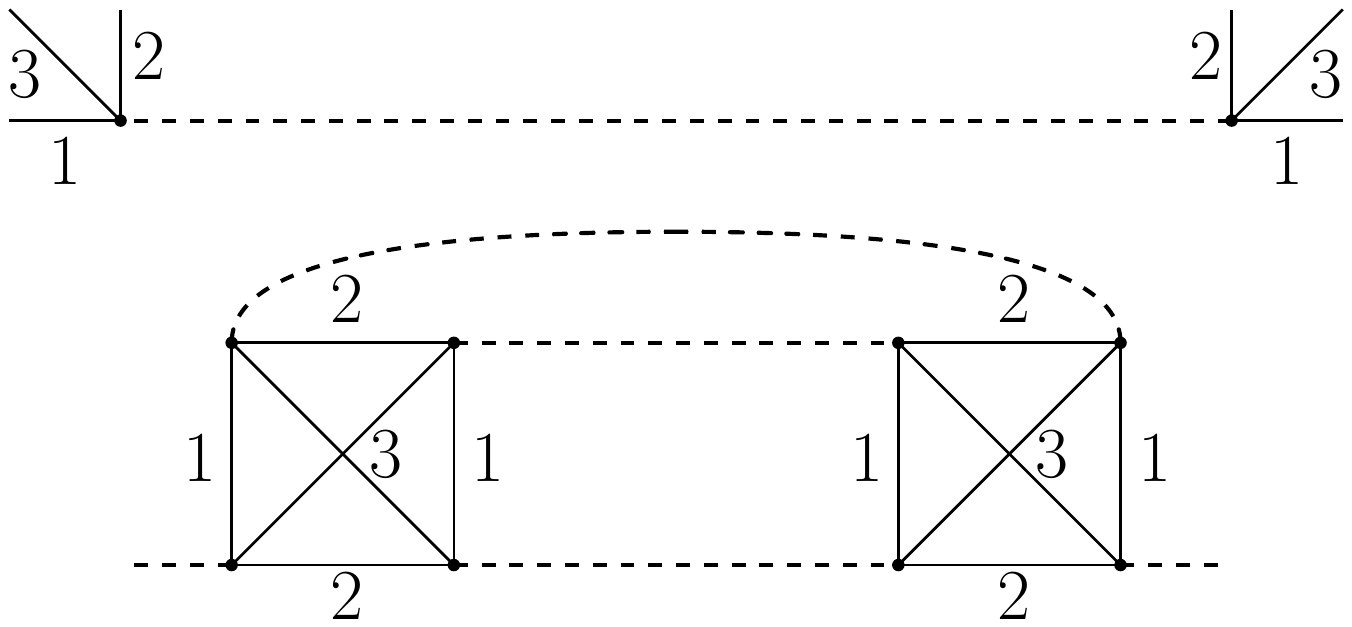} \end{array}
\end{equation}
The graph after the flip is connected. Notice that the variation of the number of faces of color 1 is not determined, but from Lemma \ref{thm:DeltaF} one has $|\Delta_{G\to G'} F_1|\leq 1$. Therefore the total variation is
\begin{equation}
\Delta_{G\to G'} F_1 + \Delta_{G\to G'} F_2 + \Delta_{G\to G'} F_3 = \Delta_{G\to G'} F_1 + 2 > 0.
\end{equation}
\end{proof}

\subsection{Proof of Theorem \ref{thm:Gmax}}

\begin{proof}
This is proved by induction on the number of bubbles. 
\begin{itemize}
\item One bubble: these are $G_1^{(1)}, G_1^{(2)}, G_1^{(3)}$.
\item Two bubbles: this is $G_2$, as can be checked by hand.
\end{itemize}

Assume that the lemma holds up to $n\geq 2$ bubbles and consider $G\in \cG_{\max}$ with $n+1$ vertices. We consider the following cases.
\begin{description}[wide=0pt]
\item[$G$ has a 2-edge-cut of color 0] we can then make a flip as follows
\begin{equation}
G = \begin{array}{c} \includegraphics[scale=.4]{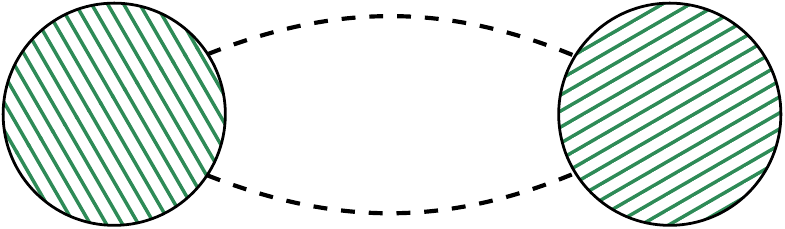} \end{array} \qquad \to \qquad G_L = \begin{array}{c} \includegraphics[scale=.4]{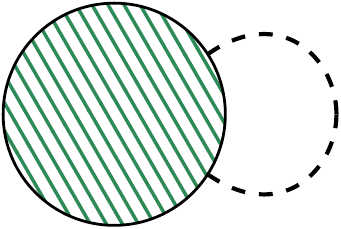} \end{array} \qquad G_R =\begin{array}{c} \includegraphics[scale=.4]{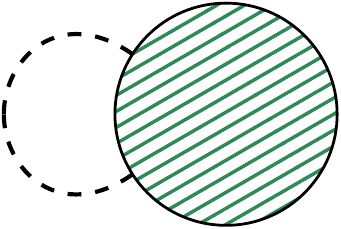} \end{array}
\end{equation}
Then 
\begin{equation}
F(G) = F(G_L) + F(G_R) - 3,
\end{equation}
so that $G\in \cG_{\max}$ requires $G_L, G_R\in \cG_{\max}$.
\begin{itemize}
\item If $b(G_L)$ and $b(G_R)$ are even, they are melonic and so is $G$ by construction.
\item If $b(G_L)$ is even and $b(G_R)$ is odd (or the other way around), it means that $G_L$ is melonic and $G_R$ is a tadpole decorated with melons. $G$ is thus also a tadpole decorated with melons.
\item If $b(G_L)$ and $b(G_R)$ are both odd, then $G$ is not in the family of the theorem. However, $b$ is even, and from Proposition \ref{thm:NumberFaces}, the number of faces of $G$ is
\begin{equation}
F(G) = \frac{3}{2}(b_L-1)+4 + \frac{3}{2}(b_R-1)+4-3 = \frac{3}{2}b(G) + 2
\end{equation}
which is less than the number of faces of a melonic graph with $b$ bubbles. We conclude that $G\not\in \cG_{\max}$.
\end{itemize}

\item[$G$ has a proper face of length 2] Then from Lemma \ref{thm:Length2} it either has a 2-edge-cut, which was treated above, or it is the melonic graph on two bubbles $G_2$.

\item[$G$ has no proper face of length 2 and two bubbles connected by more than one edge of color 0] Consider two bubbles connected by an edge and investigate whether they can be connected by other edges of color 0. Since there cannot be a proper face of length 2, a second edge of color 0 between them must be as follows
\begin{equation}
\begin{array}{c} \includegraphics[scale=.4]{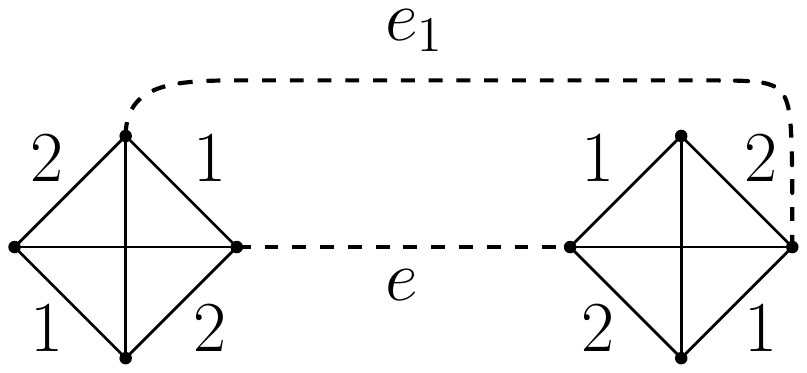} \end{array}
\end{equation}
up to color exchange. Notice an edge of color 0 between any of the remaining vertices of those bubbles would leave only two vertices to connect to the rest of the graph, i.e. $G$ would have a 2-edge-cut (or only has two bubbles but we assumed it has more).

We can thus consider that the remaining vertices of those two bubbles are connected to distinct vertices
\begin{equation}
\begin{array}{c} \includegraphics[scale=.4]{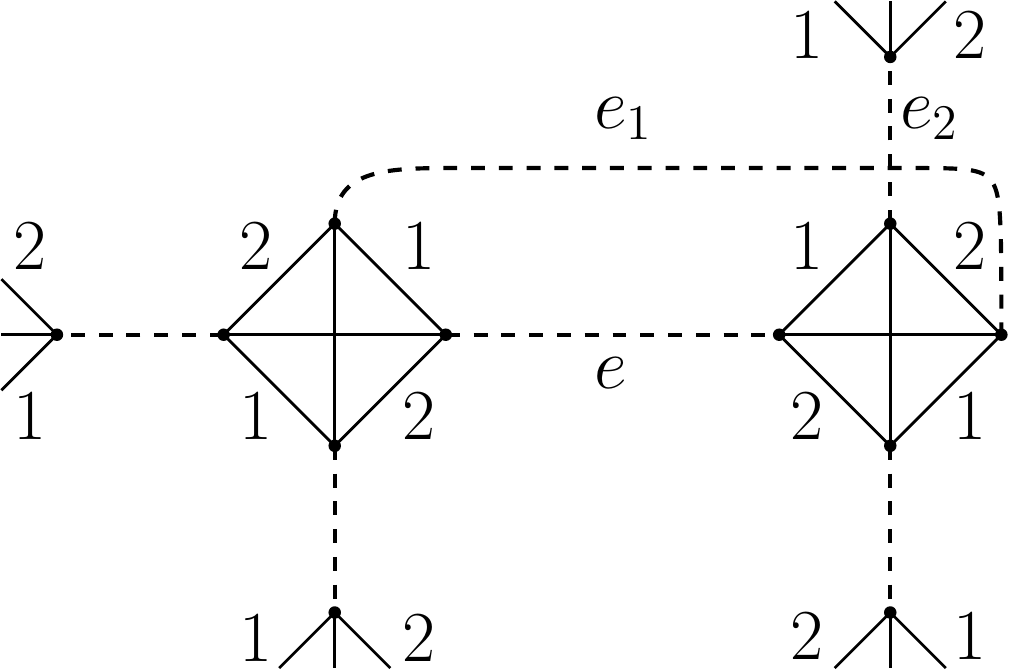} \end{array}
\end{equation}
The edges $e_1, e_2$ have two faces in common, $F_{1}^{\{e_1, e_2\}}(G) = F_2^{\{e_1, e_2\}}(G) = 1$. From the corollary \ref{thm:DeltaFTotal} we find that a flip of $\{e_1, e_2\}$ results in $\Delta_{G\to G'} F\geq 0$. It can actually vanish only when $F_3^{\{e_1, e_2\}}(G)=2$, i.e.
\begin{equation}
\begin{array}{c} \includegraphics[scale=.4]{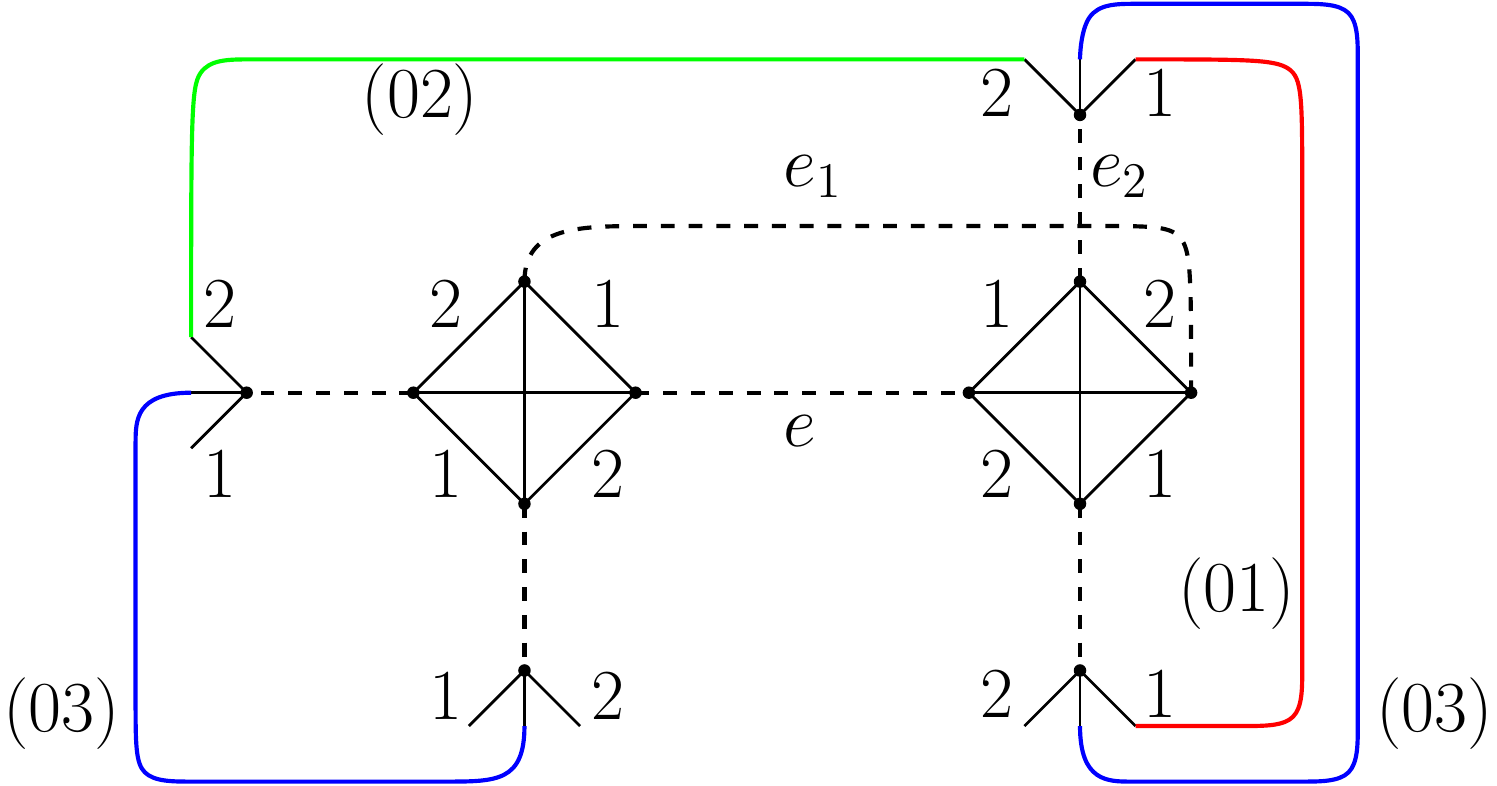} \end{array}
\end{equation}
We perform the flip as follows
\begin{equation}
\begin{array}{c} \includegraphics[scale=.4]{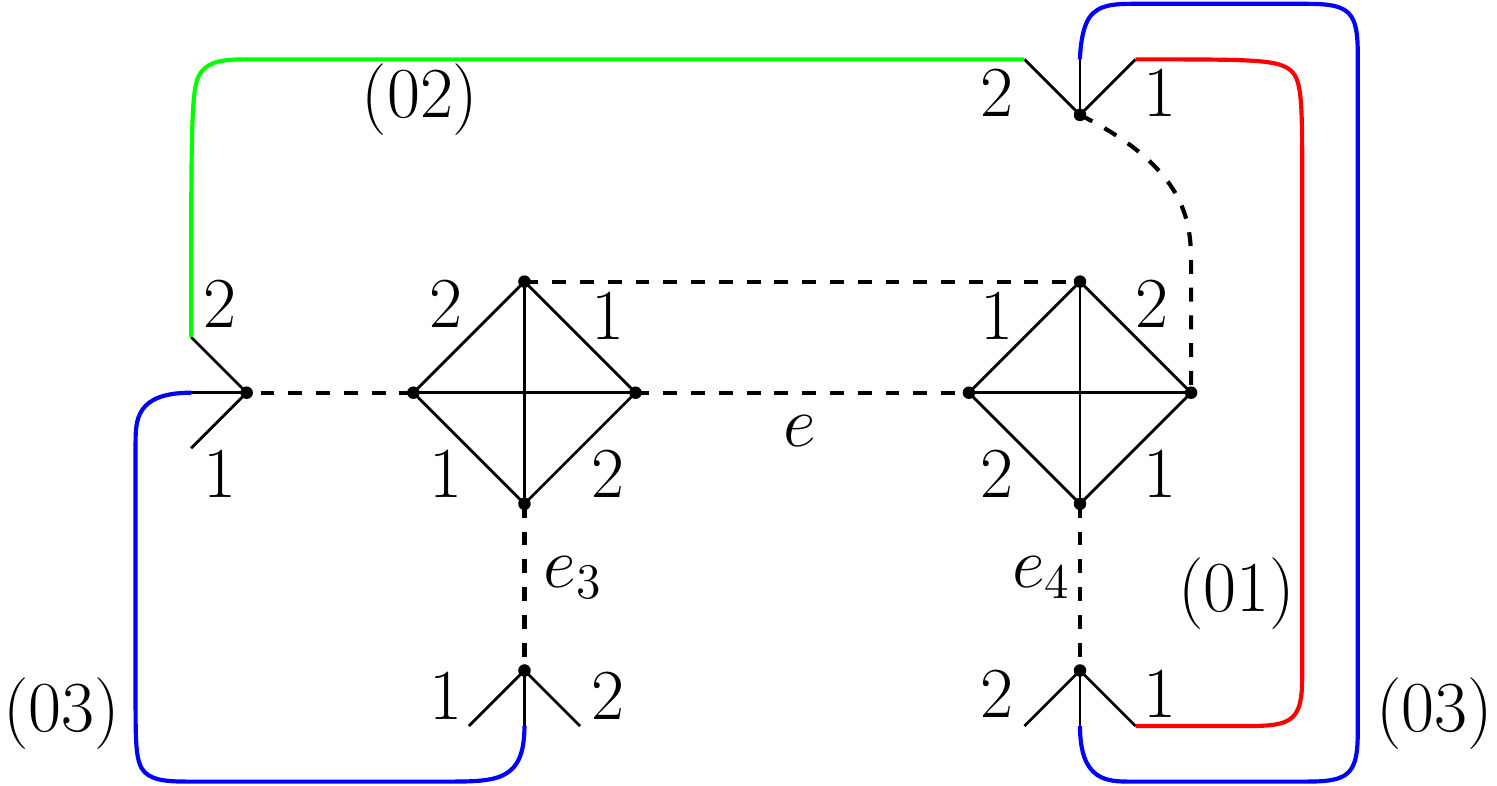} \end{array}
\end{equation}
for which $\Delta_{G\to G'} F=0$. It can now be observed that $e_3, e_4$ have two faces in common, i.e. $F_2^{\{e_3, e_4\}}(G') = F_3^{\{e_3, e_4\}}(G')=1$ and this time there is a flip which guarantees $\Delta_{G'\to G''} F>0$ since it creates a new face of color 2 and a new face of color 3
\begin{equation}
\begin{array}{c} \includegraphics[scale=.4]{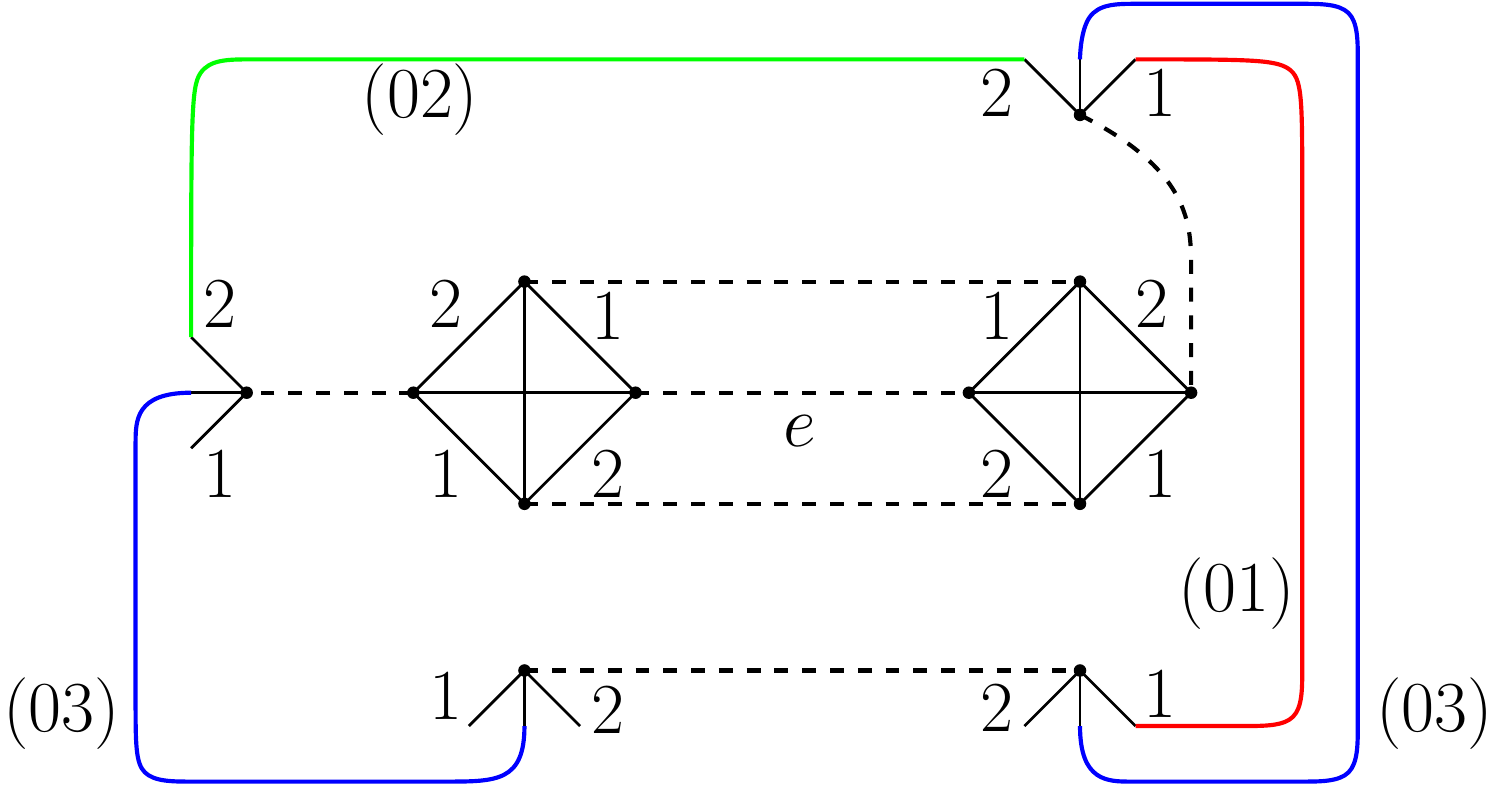} \end{array}
\end{equation}
The graph is still connected so we conclude that $G\not\in\cG_{\max}$.

\item[$G$ has no proper face of length 2 and no pair of bubbles connected by more than one edge of color 0] Consider any pair of bubbles. Their vertices are connected to vertices of other bubbles,
\begin{equation}
\begin{array}{c} \includegraphics[scale=.4]{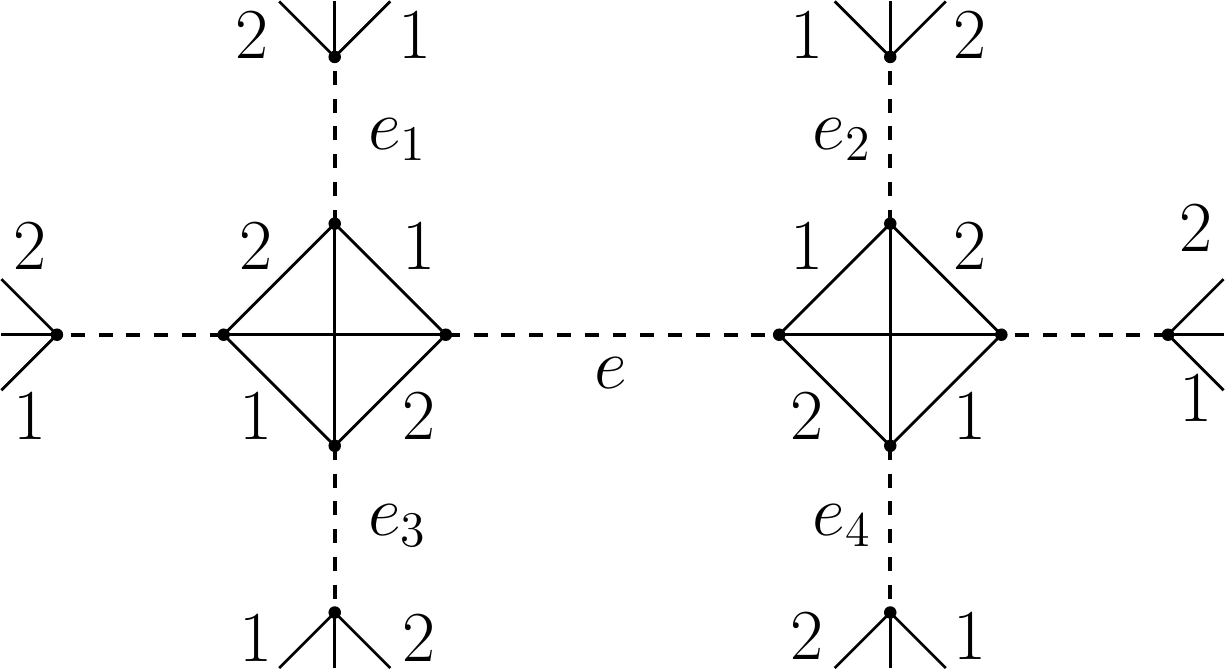} \end{array}
\end{equation}
Just like in the previous case, our strategy will be to flip $\{e_1,e_2\}$ then $\{e_3, e_4\}$. There is a flip of $\{e_1, e_2\}$ which creates a new face of color 1, i.e. $\Delta_{G\to G'} F_1=1$. As for the faces of color 2 through the flip, there are three cases,
\begin{itemize}
\item[$a)$] $\Delta_{G\to G'} F_2=1$, then given that $|\Delta_{G\to G'} F_3| \leq 1$, it comes $\Delta_{G\to G'} F>0$ and $G\not\in\cG_{\max}$.
\item[$b)$] $\Delta_{G\to G'} F_2=0$. If $\Delta_{G\to G'} F_3=1$ we can apply the case $a$ by exchanging the colors 2 and 3. If $\Delta_{G\to G'} F_3=0$ we still have $\Delta_{G\to G'} F>0$. Therefore only the case $(\Delta_{G\to G'} F_2=0, \Delta_{G\to G'} F_3=-1)$ remains to be investigated.
\item[$c)$] $\Delta_{G\to G'} F_2=-1$. Then when $\Delta_{G\to G'} F_3=0$, we can apply the case $b$ with the colors 2 and 3 exchanged. Hence only the case $(\Delta_{G\to G'} F_2=-1, \Delta_{G\to G'} F_3=-1)$ remains to be investigated.
\end{itemize}

\item[Case $b)$ $\Delta_{G\to G'} F_2=0, \Delta_{G\to G'} F_3=-1$] The paths contributing to the faces of colors 1, 2 and 3 are as follows
\begin{equation}
G = \begin{array}{c} \includegraphics[scale=.4]{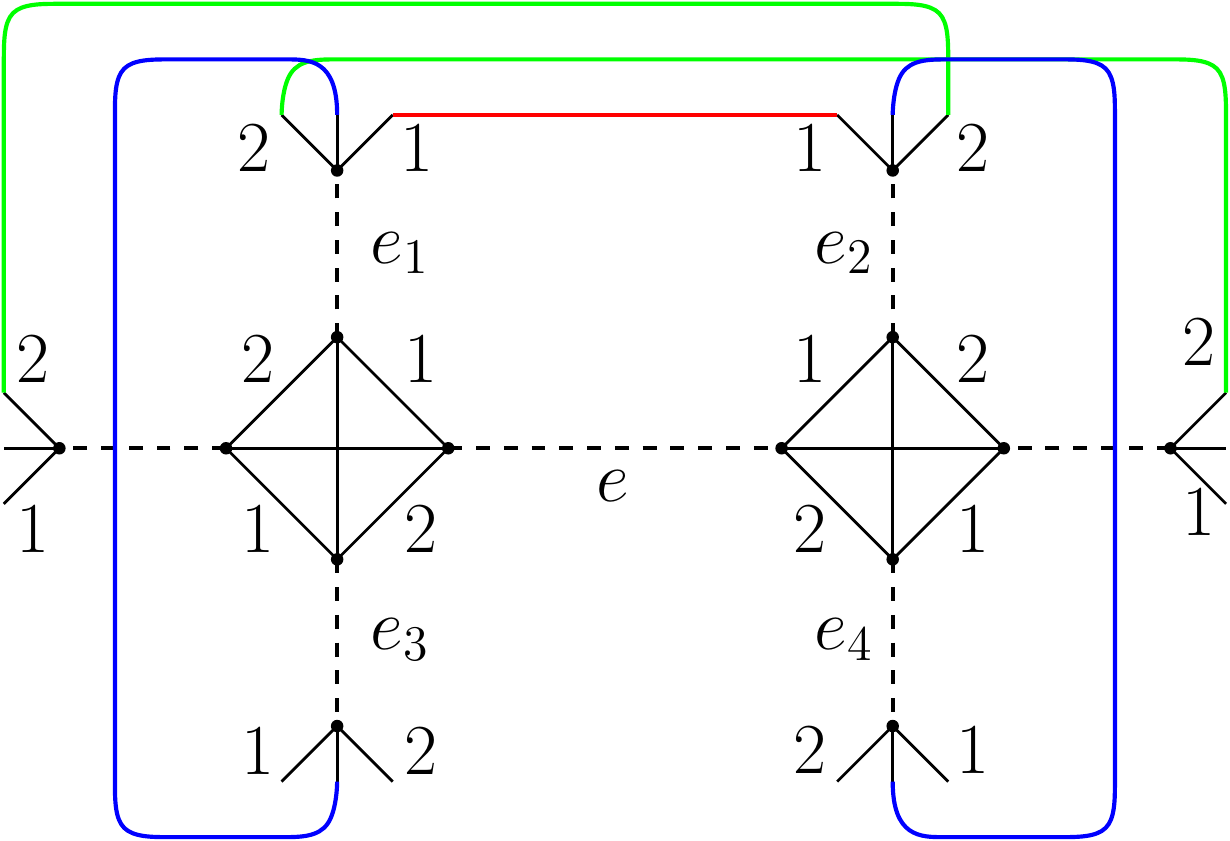} \end{array} \qquad \to \qquad
G' = \begin{array}{c} \includegraphics[scale=.4]{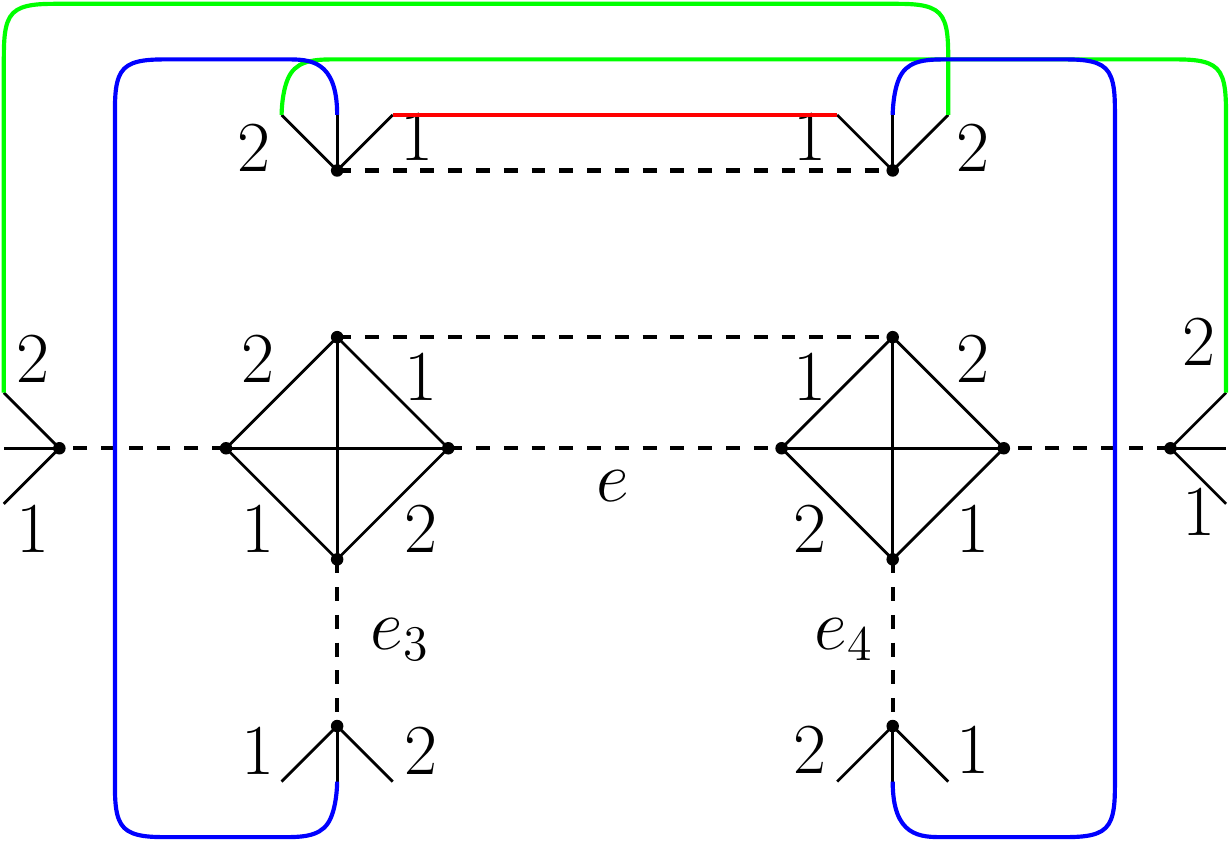} \end{array}
\end{equation}
At this stage we have $\Delta_{G\to G'} F=0$. We can then apply Lemma \ref{thm:Length2}, which is equivalent to flipping $\{e_3, e_4\}$ from $G'$ to $G''$, finally creating one more face of color 2 and one more face of color 3, 
\begin{equation}
G'' = \begin{array}{c} \includegraphics[scale=.4]{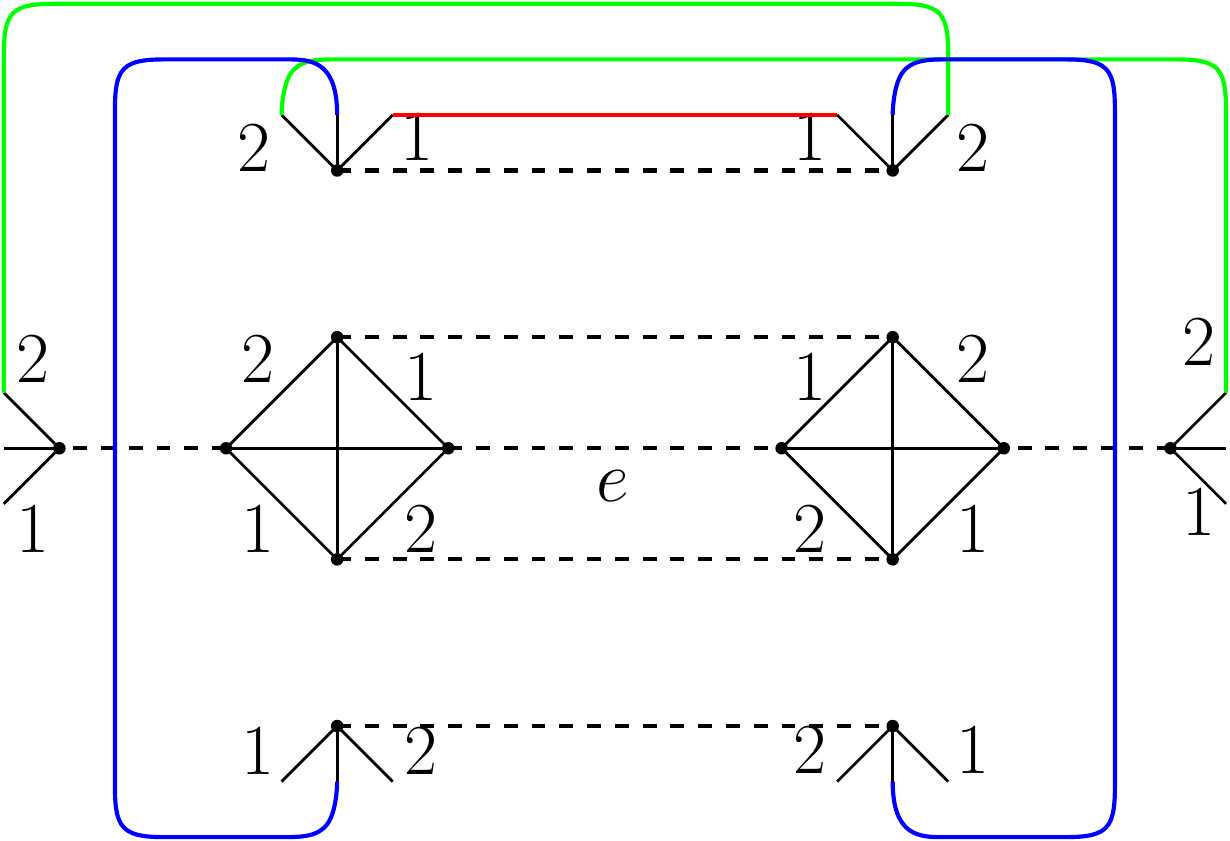} \end{array}
\end{equation}
thereby $\Delta_{G\to G''} F>0$ and $G\not\in\cG_{\max}$.

\item[Case $c)$ $\Delta_{G\to G'} F_2=-1, \Delta_{G\to G'} F_3=-1$] This is the most delicate case since $\Delta_{G\to G'} F=-1$ after the first flip meaning that a face has been lost,
\begin{equation}
G = \begin{array}{c} \includegraphics[scale=.4]{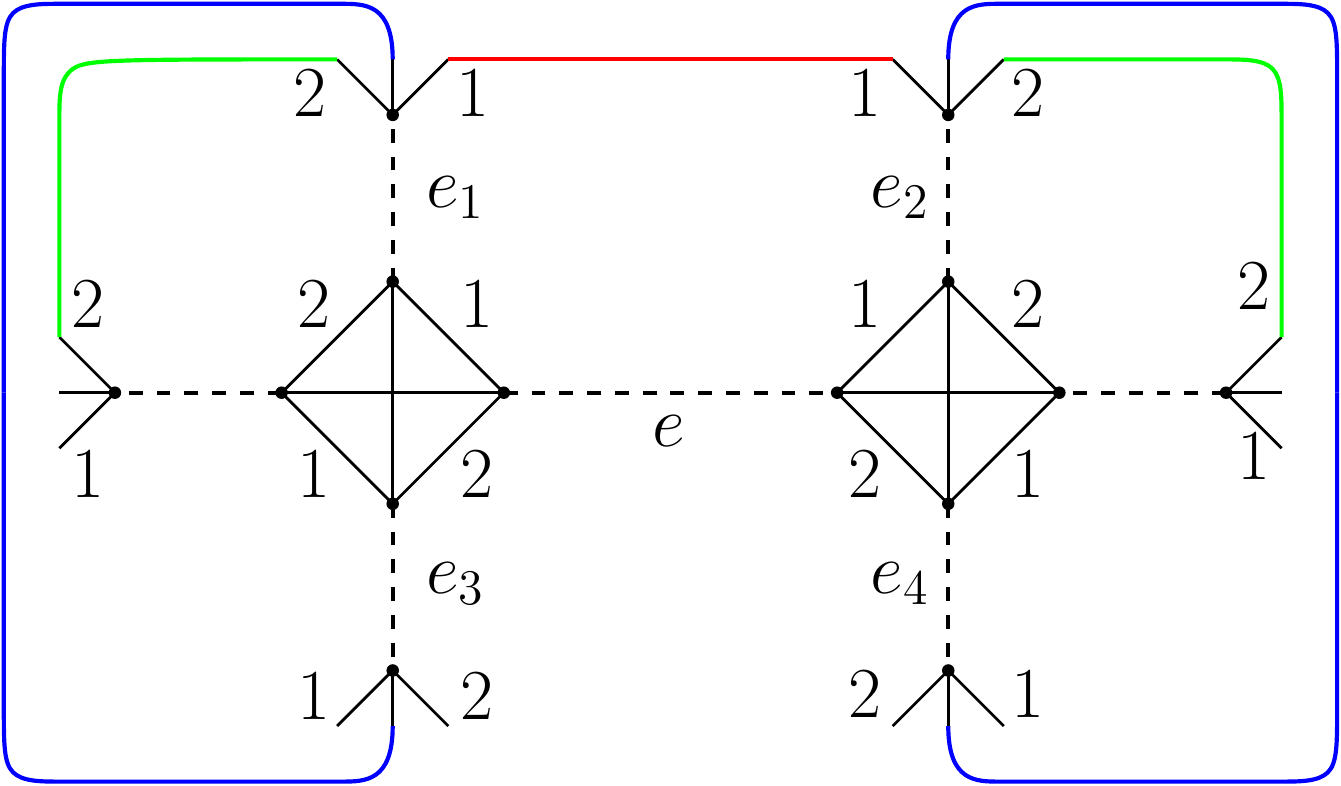} \end{array} \qquad \to \qquad
G' = \begin{array}{c} \includegraphics[scale=.4]{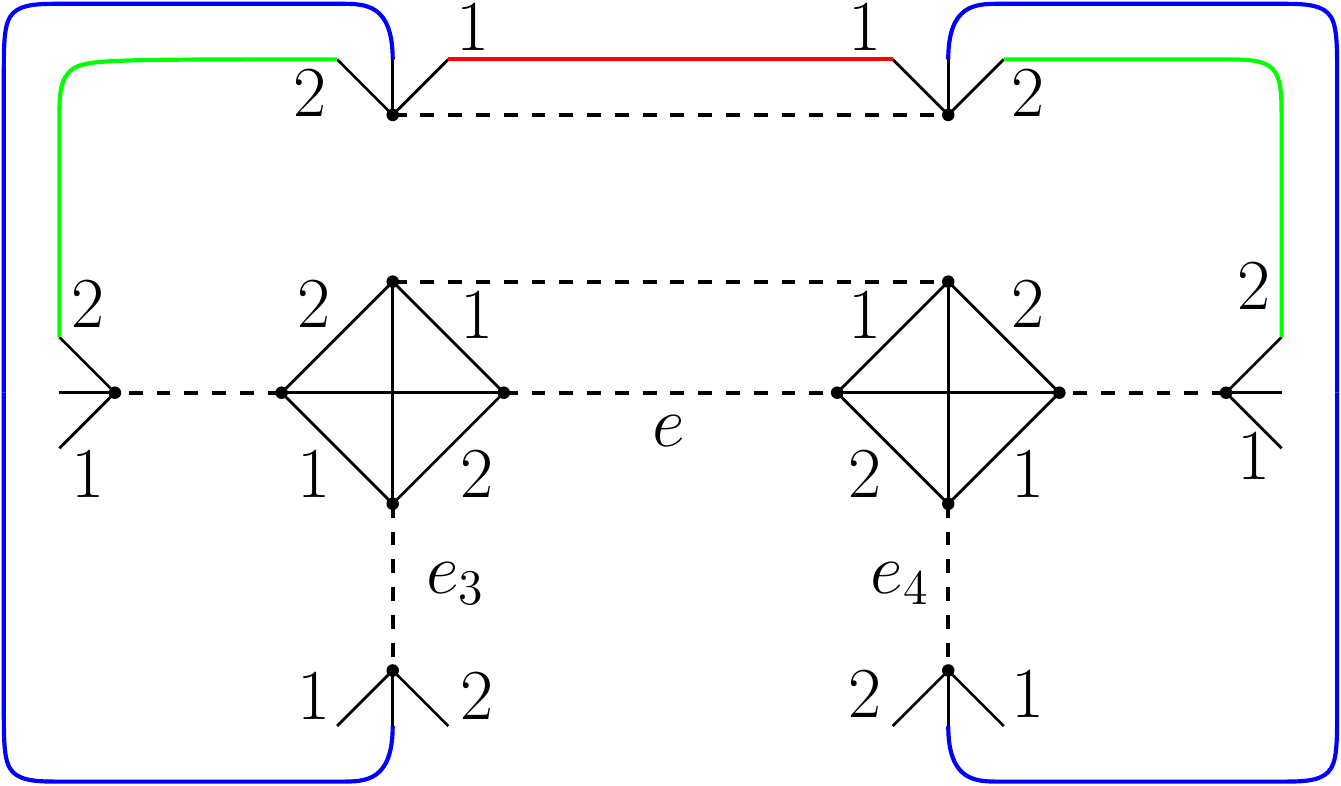} \end{array}
\end{equation}
The second flip on $\{e_3, e_4\}$ creates one more face of color 2 and one more face of color 3, so that $\Delta_{G\to G''}F\geq 0$. The only case where the total variation of the number of faces vanishes is when $F_1^{\{e_3, e_4\}}(G')=2$, i.e.
\begin{equation}
G' = \begin{array}{c} \includegraphics[scale=.4]{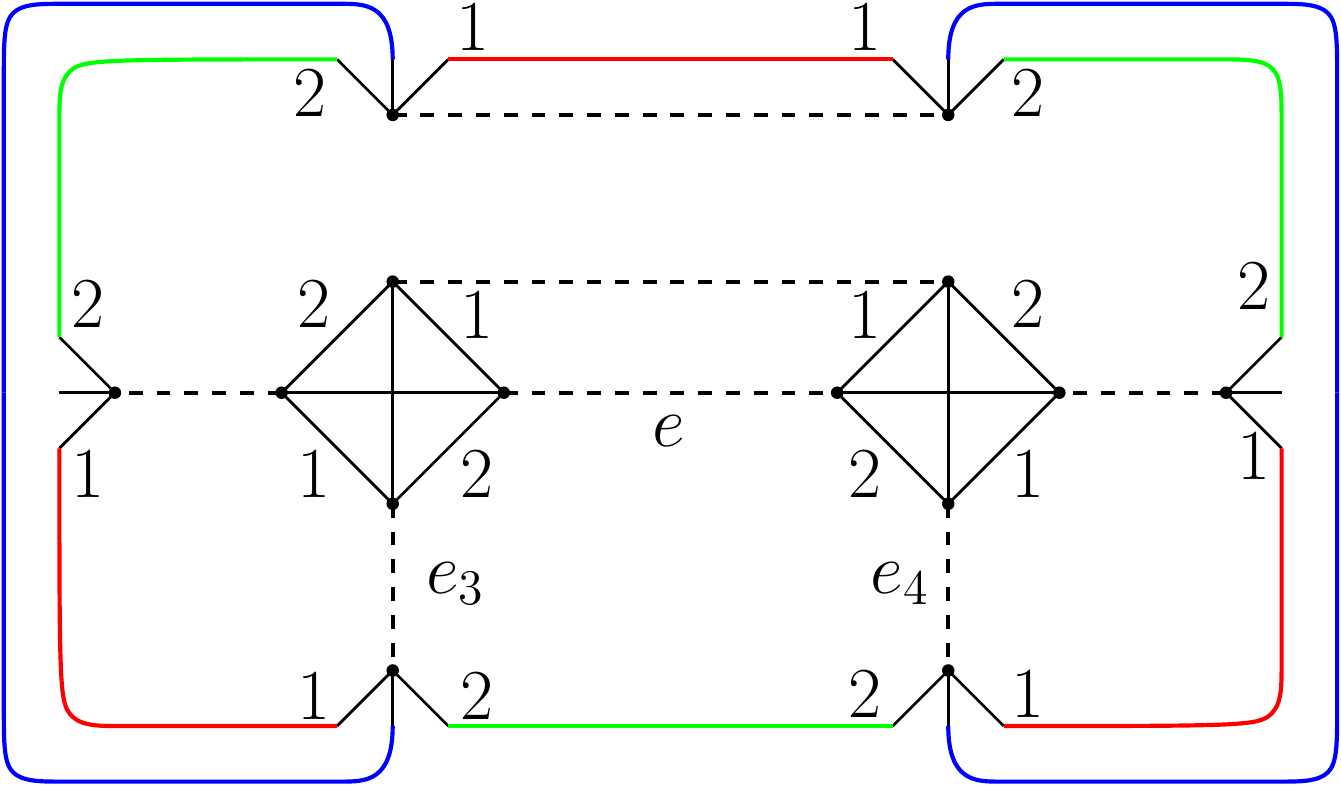} \end{array} \qquad \to \qquad
G'' = \begin{array}{c} \includegraphics[scale=.4]{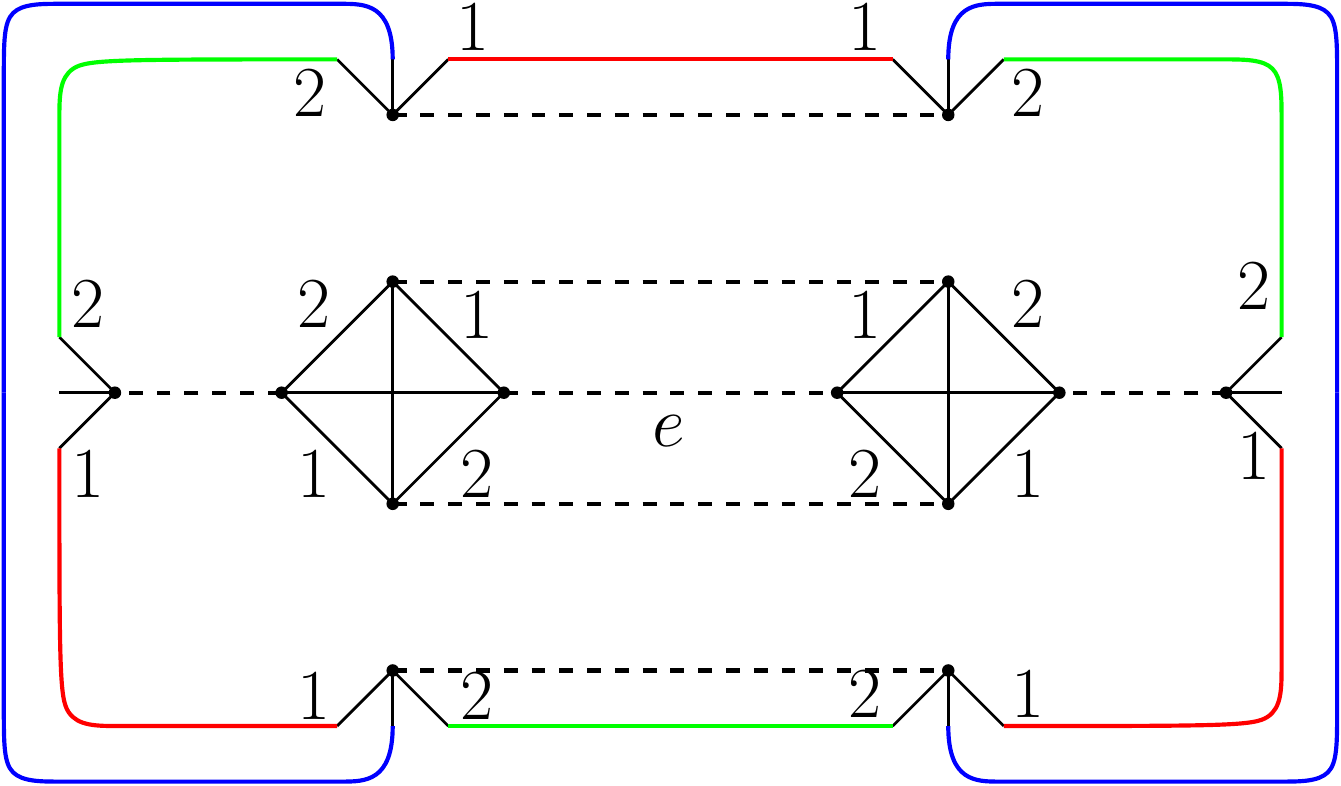} \end{array}
\end{equation}
One then performs a melonic reduction, i.e. removing the melonic dipole which has been created by the two flips,
\begin{equation}
G'' \qquad \to \qquad \tilde{G} = \begin{array}{c} \includegraphics[scale=.4]{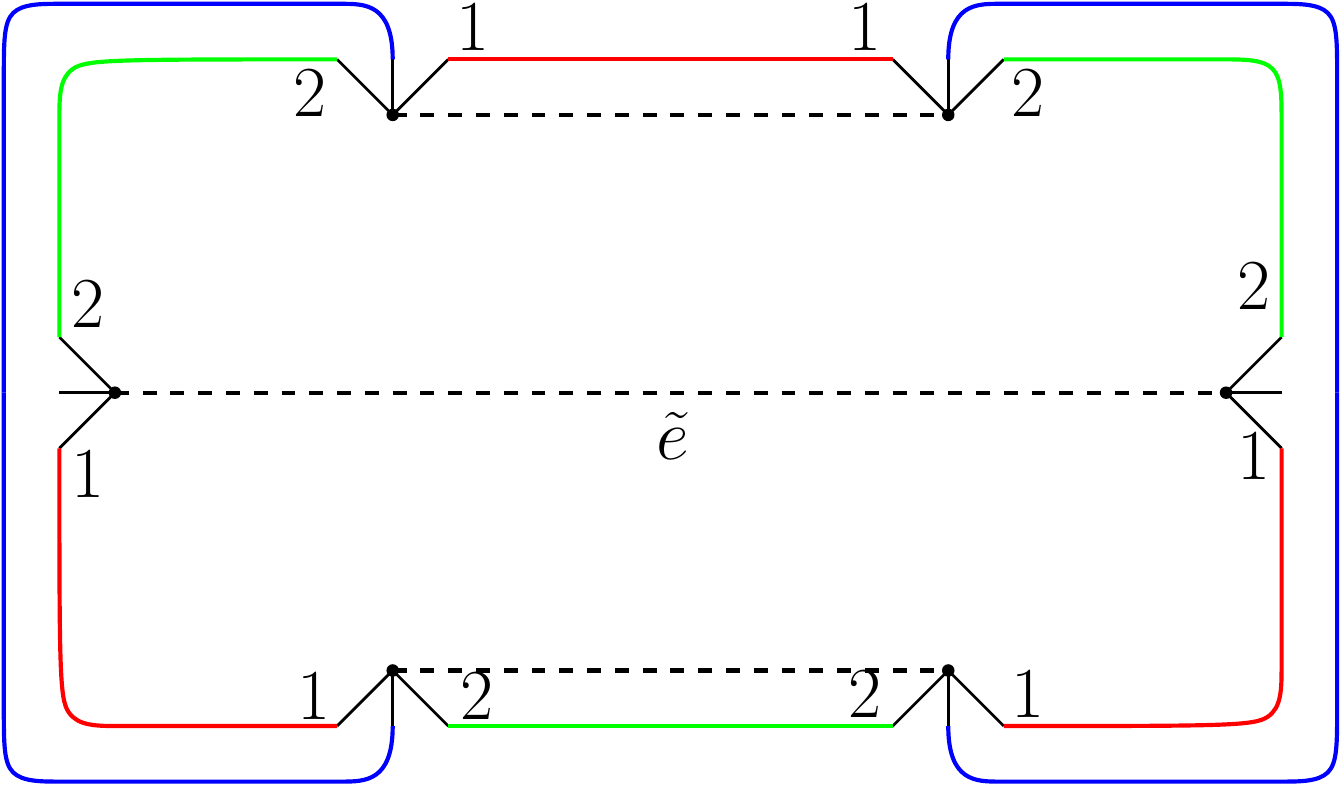} \end{array}
\end{equation}
$\tilde{G}$ has two bubbles less than $G$ and from our induction hypothesis it has to be melonic or a tadpole decorated with melons. We will show however that it cannot be, because it has no proper faces of length 2.

Let us re-analyze the sequence of transformations $G\to G' \to G'' \to \tilde{G}$ and pay attention to the creation of proper faces of length 2. We start with $G$ which has none. $G'$ only has one. Indeed, the faces whose lengths can be shortened through the flip are only the one face of color 1 which goes along $e_1$ and $e_2$ (the other faces along $e_1$ and $e_2$ are merged and thus lengthened, while the other faces of $G$ are unaffected). The flip might therefore create a second proper face of length 2 as follows,
\begin{equation}
\begin{array}{c} \includegraphics[scale=.4]{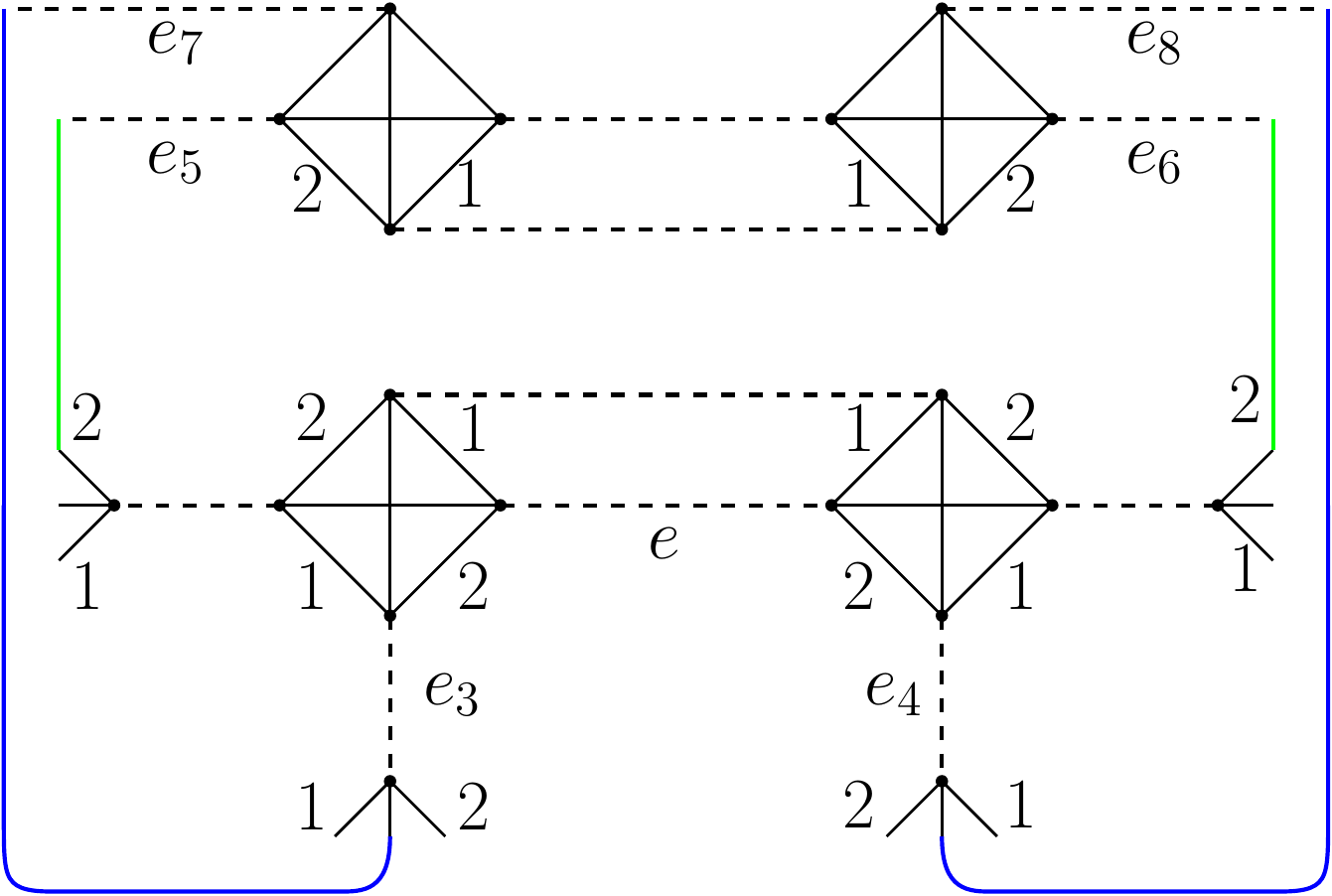} \end{array}
\end{equation}
From Lemma \ref{thm:Length2}, the pair of edges $\{e_5, e_6\}$ should form a 2-edge-cut which it is not (same for $\{e_7, e_8\}$).

The same argument applies to the transformation $G'\to G''$, from which we conclude that the only proper faces of length 2 of $G''$ are the three of its newly created melonic dipole. Finally, the reduction $G''\to \tilde{G}$ does not create proper faces of length 2. Indeed, assume that a proper face of length 2 and color 2 was created,
\begin{equation}
\tilde{G} = \begin{array}{c} \includegraphics[scale=.4]{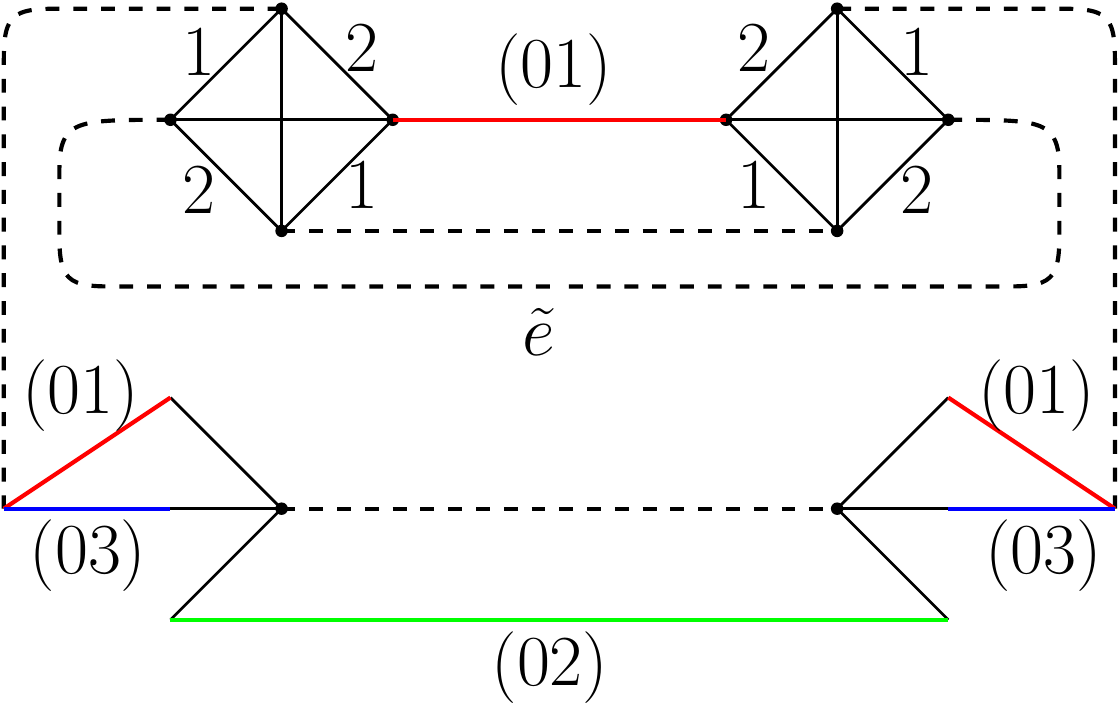} \end{array}
\end{equation}
By adding the melonic dipole back on $\tilde{e}$ and flipping the edges back, it appears that $G$ must have a proper face of length 2 which is a contradiction. Therefore $\tilde{G}\not\in\cG_{\max}$ and so is not $G$.
\end{description}
\end{proof}

\section*{Acknowledgements}

This research was supported by the ANR MetACOnc project ANR-15-CE40-0014.


\end{document}